\documentclass[
reprint,
superscriptaddress,
amsmath,
amssymb,
aps,
prl
]{revtex4-2}
\usepackage{graphicx}
\usepackage{dcolumn}
\usepackage{bm}
\usepackage{hyperref}
\usepackage{amsmath,amssymb,amsfonts,amsthm}
\usepackage{graphicx}
\usepackage{hyperref}
\usepackage{bm}
\usepackage{mathtools}
\usepackage{cancel}
\usepackage{xcolor}
\usepackage{xcolor}
\usepackage[mathscr]{euscript}
\usepackage{mathrsfs}
\newtheorem{theorem}{Theorem}

\newtheorem{definition}{Definition}

\newcommand{\tr}{\operatorname{tr}}
\newcommand{\Ric}{\operatorname{Ric}}
\newcommand{\Scal}{\operatorname{Scal}}

\newcommand{\Rad}{\operatorname{Rad}}

\newcommand{\ubar}{\underline{u}}
\newcommand{\Lbar}{\underline{L}}
\newcommand{\Hbar}{\underline{H}}
\newcommand{\chibar}{\underline\chi}
\newcommand{\gslash}{\cancel{g}}

\begin{document}

\title{Boundary-Geometry-Driven Black Hole Formation in Vacuum}

\author{Dipanjan Dey}
\affiliation{Beijing Institute of Mathematical Sciences and Applications,
Yau Mathematical Sciences Center, Tsinghua University, Beijing 100084, China}
\author{Puskar Mondal}
\affiliation{Beijing Institute of Mathematical Sciences and Applications,
Yau Mathematical Sciences Center, Tsinghua University, Beijing 100084, China}
\author{Shing-Tung Yau}
\affiliation{Beijing Institute of Mathematical Sciences and Applications,
Yau Mathematical Sciences Center, Tsinghua University, Beijing 100084, China}

\date{\today}

\begin{abstract}
We identify a boundary-geometry-driven mechanism for the dynamical formation of marginally outer trapped surfaces (MOTSs) in vacuum general relativity. Mild three-dimensional anisotropies evolve inside a compact Cauchy domain whose effective isotropic thickness remains controlled, while the generalized boundary mean curvature can increase during either contracting or expanding boundary evolution. This drives the boundary across Yau's geometric threshold, forcing MOTS formation from initially untrapped data. We further interpret the characteristic shear construction of Ref.~\cite{MondalYau2026} as the null manifestation of the same anisotropic vacuum dynamics. The result provides a purely vacuum physical realization of MOTS formation through global geometric effects, without invoking a short-pulse concentration mechanism for gravitational radiation.
\end{abstract}

\maketitle


\noindent One of the central problems in the mathematical theory of black holes is to understand not only their existence, but their dynamical formation from regular initial data. Penrose’s incompleteness theorem asserts that a spacetime satisfying the null energy condition and admitting a noncompact Cauchy hypersurface is future null geodesically incomplete if it contains a closed spacelike two-surface whose two future-directed null expansions, $(\tr\chi)$ and $(\tr\chibar)$, are both negative. Such a surface is known as a \emph{closed trapped surface} \cite{Penrose1965,HawkingEllis}. In this sense, the theorem is conditional: it assumes the existence of the trapped surface that signals gravitational collapse, but does not explain how such a surface forms from regular initial data. The basic dynamical problem is therefore to identify sufficiently general untrapped initial configurations whose future evolution necessarily develops closed trapped surfaces.

\noindent Schoen and Yau \cite{SchoenYau1983} first established an initial-data criterion for the existence of a marginally outer trapped surface (MOTS), later extended by Yau to a broader setting allowing even negative energy density \cite{Yau2002}. Their global criterion couples the bulk and boundary geometry through the $H$-radius, which measures the isotropic thickness of a compact domain, and guarantees a closed surface satisfying
$
\tr\chi=0,~~\tr\underline{\chi}<0.
$

Christodoulou obtained the first nonsymmetric dynamical formation of trapped surfaces through the \textit{short-pulse method}, whose hierarchy of large and small initial-data components is propagated by the Einstein evolution up to controlled errors \cite{Christodoulou2009}. This framework was later refined through alternative hierarchies and matter couplings; see, e.g., \cite{Kl-Rod,LKR,chen,A19,AL17,A-M-Y,athanasiou2025semi}. In Christodoulou's setting, the incoming null data are trivial and the interior is Minkowskian. By contrast, Ref.~\cite{MondalYau2026} constructs an open set of smooth, asymptotically flat vacuum data with nontrivial interior geometry and no initial trapped surfaces or MOTSs, for which Yau's boundary condition arises dynamically. Here we give the first explicit realization of Yau’s boundary-geometry mechanism for the dynamical formation of marginally outer trapped surfaces (MOTSs). We show that, within a compact Cauchy domain, mild three-dimensional anisotropies may evolve while the effective isotropic thickness remains controlled, yet the generalized boundary mean curvature can increase under either contracting or expanding boundary evolution; once it crosses Yau’s geometric threshold, a MOTS is necessarily forced to form from initially un-trapped data.

  
\begin{definition}
    Let $\Gamma$ be a simple closed curve in $M$ that bounds a disk in $M$. Consider $N_r(\Gamma)$ as the set of points within a distance $r$ of $\Gamma$. The H-radius of $M$ with respect to $\Gamma$ is defined as:
\begin{eqnarray}
    \text{Rad}(M, \varGamma): = \sup \{ r : \text{dist}(\varGamma, \partial M) > r, \ \varGamma \text{ does not bound}\nonumber\\
    \text{ a disk in } N_r(\varGamma) \}.\nonumber
\end{eqnarray}
Now, the H-radius of $M$, denoted by $\text{Rad}(M)$, defined as $\text{Rad}(M): = \sup \left\{ \text{Rad}(M, \varGamma) : \varGamma \text{ as above} \right\}$.
\end{definition}

For example, the H-radius of a ball of radius $R$ in $\mathbb{R}^3$ is: $\text{Rad}(M)=\frac R2$. The H-radius of a solid cylinder ($D^2\times (+L,~-L)$) of radius $R$ and length $2L$ in $\mathbb{R}^3$ is: $\text{Rad}(M)=\min\left(\frac R2,L\right)$.  


\noindent In Yau's \cite{Yau2002} theorem, the key geometric entities are the Schoen-Yau $H$-radius of the compact domain and the generalized mean curvature ($c$) of the boundary ($\partial M$).
To state the theorem geometrically, let \(\mathcal{M}_t\) be a spacelike Cauchy hypersurface in \((\mathcal{M},\mathrm{g})\), and let \(M_t\subset\mathcal{M}_t\) be a compact domain with smooth boundary \(S_t=\partial M_t\). The induced initial data \((g,k)\) satisfy the Einstein constraint equations.
Let \(e_T\) be the future-directed unit normal to \(\mathcal{M}_t\), \(e_S\) the outward unit normal to \(S_t\) within \(\mathcal{M}_t\), and \({e_A},{A=1,2}\) an orthonormal frame tangent to \(S_t\). Define
$
s_{AB}
=
\left\langle
\nabla_{e_A}e_S,e_B
\right\rangle,
\
k_{AB}
=
\left\langle
\nabla_{e_A}e_T,e_B
\right\rangle,
$
where \(\nabla\) is the Levi-Civita connection of \(\mathrm{g}\). If \(\Sigma\) is the induced metric on \(S_t\), then
$
H
=
\Sigma^{AB}s_{AB},~
\kappa
=
\Sigma^{AB}k_{AB}.
$
Thus, \(H\) is the mean curvature of \(S_t\subset\mathcal{M}_t\), while \(\kappa=\operatorname{tr}_{S_t}k\). In vacuum, Yau's theorem implies that if
\begin{equation}
\operatorname{Rad}(M_t)
\geq
\frac{3\pi}{2c},
\qquad
c:=
\min_{S_t}
\left(
H-|\kappa|
\right)>0,
\end{equation}
then $(M_t)$ contains a marginally outer trapped surface.

\noindent In Ref.~\cite{MondalYau2026}, the latter two authors show that Yau's boundary condition can arise dynamically from an open set of smooth, asymptotically flat vacuum initial data containing neither trapped surfaces nor MOTSs. On the initial slice \(\mathcal{M}_{-a}\), with \(a\gg 1\), they decompose
$
\mathcal{M}_{-a}
=
M_1\cup M_2\cup M_{\mathrm{ext}},
\qquad
\partial M_1=S_{-a,0},
$
where \(M_1\) is a compact interior domain, \(M_2\) is generated by the semiglobal characteristic development, and \(M_{\mathrm{ext}}\) is the far exterior; see Fig.~\ref{evol1}. The interior is further decomposed as
$
M_1
=
\widetilde{M}_1
\cup
\left(
M_1\setminus\widetilde{M}_1
\right),
$
where \(\widetilde{M}_1\) is a compact core and \(M_1\setminus\widetilde{M}_1\) is a collar in which the characteristic and Cauchy data are smoothly matched.
The evolution remains regular up to a finite proper time \(t_*>0\), at which a compact domain \(M_{t_*}\subset\mathcal{M}_{t_*}\) contains a MOTS. The proof controls \(H-|\kappa|\) along the evolving boundary of the causal future of \(M_1\), while forcing the \(H\)-radius of
$J^+(M_1)\cap\mathcal{M}_t$
to increase. At \(t=t_*\), these estimates establish Yau's boundary condition on
$\partial J^+(M_1)\cap\mathcal{M}_{t_*},$
and hence imply the existence of a MOTS in \(M_{t_*}\). Here \(M_t=\Phi_t(M_1)\), where \(\Phi_t\) is the flow generated by the global time vector field \(\partial_t\).
Ref.~\cite{MondalYau2026}, however, does not isolate the internal mechanism driving the formation of MOTS. In particular, it does not quantify the evolution of the geometric anisotropy or its effect on the \(H\)-radius and boundary geometry of \(M_t\). In this article, our purpose is not to claim that MOTS formation itself forces the
anisotropy conditions derived below, but to identify intrinsic sufficient
conditions within this vacuum development under which the evolving
boundary reaches Yau's threshold.

\noindent We study the proper-time evolution of a compact spacelike domain and identify an intrinsic geometric mechanism for MOTS formation. The key observation is that control of the trace-free part of the three-dimensional Ricci tensor, \(\Ric^{(3)}\), prevents substantial collapse of the \(H\)-radius, while the generalized mean curvature of the boundary increases along the evolution. Their combined effect forces the evolving boundary across the Yau threshold and guarantees the appearance of a MOTS in the interior. This yields a description of horizon formation formulated entirely in terms of the geometry of the evolving compact domain and reveals how the characteristic radiation mechanism of Ref.~\cite{MondalYau2026} is encoded in the evolving geometry of the compact domain.
This viewpoint is distinct from that of Ref.~\cite{MondalYau2026}, where MOTS formation is obtained through a characteristic evolution. There, semi-global well-posedness up to a time \(t_*\) is established by imposing conditions on the freely prescribed shear \(\hat{\chibar}\) of the ingoing null second fundamental form on \(\ubar=0\). The resulting gravitational radiation controls \(H-|\kappa|\) along the outgoing null boundary \(\partial J^+(M_1)\), driving
\(c=\min_{\partial J^+(M_1)}\bigl(H-|\kappa|\bigr)\)
toward a nearly uniform profile that leads to the formation of a MOTS by satisfying Yau's boundary condition at \(t=t_*\), on the boundary 
\(\partial J^+(M_{t=0})\cap\mathcal M_{t_*}\).
In this paper, for the mathematical analysis relevant to the semi-global existence theory, we adapt the approach similar to that of Ref.~\cite{MondalYau2026}, which provides the smooth evolution up to the formation of the first MOTS.

We briefly recall the double-null gauge used throughout; see Refs.~\cite{MondalYau2026,CBG}. Let \((\mathcal M,\mathrm g)\) be a spacetime admitting smooth optical functions
$(u,\ubar):\mathcal M\longrightarrow \mathbb R^2$
whose level sets
$H_u:={u=\mathrm{const}},
\qquad
\Hbar_{\ubar}:={\ubar=\mathrm{const}}$
are smooth outgoing and incoming null hypersurfaces, respectively. Whenever \(H_u\cap\Hbar_{\ubar}\neq\varnothing\), we set
$S_{u,\ubar}:=H_u\cap\Hbar_{\ubar}.$
We assume that each \(S_{u,\ubar}\) is a smooth, embedded, spacelike two-sphere, and denote its induced Riemannian metric by \(\gslash\).
In double-null coordinates \((u,\ubar,\theta^1,\theta^2)\), the spacetime metric takes the form
\begin{eqnarray}
\nonumber \mathrm g=
-2\Omega^2
\left(
\mathrm du\otimes\mathrm d\ubar
+
\mathrm d\ubar\otimes\mathrm du
\right)
\\+
\gslash_{AB}
\left(
\mathrm d\theta^A-b^A\mathrm du
\right)
\otimes
\left(
\mathrm d\theta^B-b^B\mathrm du
\right),
\end{eqnarray}
where \(\Omega\) is the null lapse, \(b=b^A\partial_{\theta^A}\) is the angular shift, and
$\gslash_{AB}=\gslash\left(
\frac{\partial}{\partial\theta^A},
\frac{\partial}{\partial\theta^B}
\right).$
The associated null vector fields are
$e_3=
\Omega^{-1}
\left(
\frac{\partial}{\partial u}
+
b^A\frac{\partial}{\partial\theta^A}
\right),
\qquad
e_4
=\Omega^{-1}\frac{\partial}{\partial\ubar},$
while \(e_A=\partial_{\theta^A}\), \(A=1,2\), span \(TS_{u,\ubar}\).
Let \(D\) denote the Levi-Civita connection of \(\mathrm g\). The Ricci coefficients and curvature components relative to this frame are defined by
$\Gamma_{\lambda\mu\nu}
\mathrm g\left(
e_\lambda,D_{e_\mu}e_\nu
\right),~
R_{\alpha\beta\mu\nu}
\mathrm{Riem}=\left(
e_\alpha,e_\beta,e_\mu,e_\nu
\right).$
The Einstein and Bianchi equations then decompose into the standard null structure and null Bianchi equations. We use the conventions and explicit formulas of Ref.~\cite{MondalYau2026}.

We next introduce the spacelike foliation induced by the double-null structure. Define
$t:=u+\ubar,$
and, for each \(t\in\mathbb R\), let
$\mathcal M_t:={p\in\mathcal M:u(p)+\ubar(p)=t}.$
We assume that, throughout the time interval under consideration, each \(\mathcal M_t\) is a smooth, asymptotically flat spacelike Cauchy hypersurface for the relevant portion of \((\mathcal M,\mathrm g)\). Thus, \({\mathcal M_t}\) defines a spacelike foliation compatible with the double-null gauge.
Now, let \(M_t\subset\mathcal M_t\) be a compact domain with smooth boundary
$S_t:=\partial M_t.$
We assume that \(M_{t_0}\) contains no MOTS and that no MOTS occurs in \(M_t\) for \(t\in[t_0,t_*)\), where \(t_*=u_*+\ubar_*\) denotes the first MOTS formation time. We further assume that the double-null leaves contained in \(M_t\) remain untrapped before \(t_*\):
$\tr\chi>0,
\qquad
\tr\chibar<0,
\qquad
t\in[t_0,t_*).$
Here \(e_T\) is the future-directed unit normal to \(\mathcal M_t\), \(e_S\) is the outward unit normal to the relevant two-surface within \(\mathcal M_t\), and
$L=e_T+e_S,
\qquad
\Lbar=e_T-e_S$
are the outgoing and ingoing null normals. If \(\Sigma\) denotes the induced metric and \({e_A}_{A=1,2}\) is a local \(\Sigma\)-orthonormal frame, then
$\tr\chi=
\Sigma^{AB}\mathrm g(D_{e_A}L,e_B),
\qquad
\tr\chibar
=\Sigma^{AB}\mathrm g(D_{e_A}\Lbar,e_B).$

We assume that the boundary \(S_t\) is transported along \(e_T\). Let \(T\) denote proper time along the integral curves of \(e_T\), so that \(\partial_T=e_T\). The local rate of change of the area element \(d\mathcal A\) is
\begin{equation}
\kappa
:=
\frac{1}{d\mathcal A}\frac{\partial d\mathcal A}{\partial T}=
\Sigma^{AB}\mathrm g(D_{e_A}e_T,e_B).
\label{collapsedef0}
\end{equation}
Accordingly,
$\kappa<0$
describes pointwise contraction of \(S_t\), whereas
$\kappa\geq0$
describes point-wise expansion, including the marginal case \(\kappa=0\), for which the local area element is stationary along the \(e_T\)-flow.
 On the boundary $\partial M_t$, Equation~\eqref{collapsedef0} may be expressed in terms of the null expansions as,
\begin{equation}
\kappa=
\frac{1}{2}\left(\tr\chi+\tr\chibar\right)
\begin{cases}
<0, & \text{contraction phase},\\
>0, & \text{expansion phase},\\
=0, & \text{area-preserving shear phase}
\end{cases}
\label{collapsedef}
\end{equation}
Moreover, since
$H=\frac{1}{2}\left(\tr\chi-\tr\chibar\right),$
the generalized mean curvatures during these three phases satisfy
\begin{equation}
c(t):=\min_{\partial M_t}\bigl(H-|\kappa|\bigr)=
\begin{cases}
=\min_{\partial M_t}\tr\chi, & \text{for $\kappa <0$},\\
=-\max_{\partial M_t}\tr\chibar, & \text{for $\kappa >0$},\\
=\min_{\partial M_t}H, & \text{for $\kappa =0$}.
\end{cases}
\label{collapsedef}
\end{equation}

\noindent Assuming the hypotheses of the Schoen--Yau criterion, the absence of a MOTS in \(M_{t_0}\) implies, by contraposition,
$c(t_0)
<
\frac{3\pi}{2\operatorname{Rad}(M_{t_0})}.$
Thus, for the criterion to be reached at a later time \(t_*\), the dimensionless quantity
$\mathcal Q(t):=\operatorname{Rad}(M_t)c(t)$
must increase from a value below \(3\pi/2\) to one satisfying
$\mathcal Q(t_*)\geq\frac{3\pi}{2}.$

We now provide the main theorem and the sketch of the proof here. 

\begin{figure} 
 \includegraphics[scale=0.4]{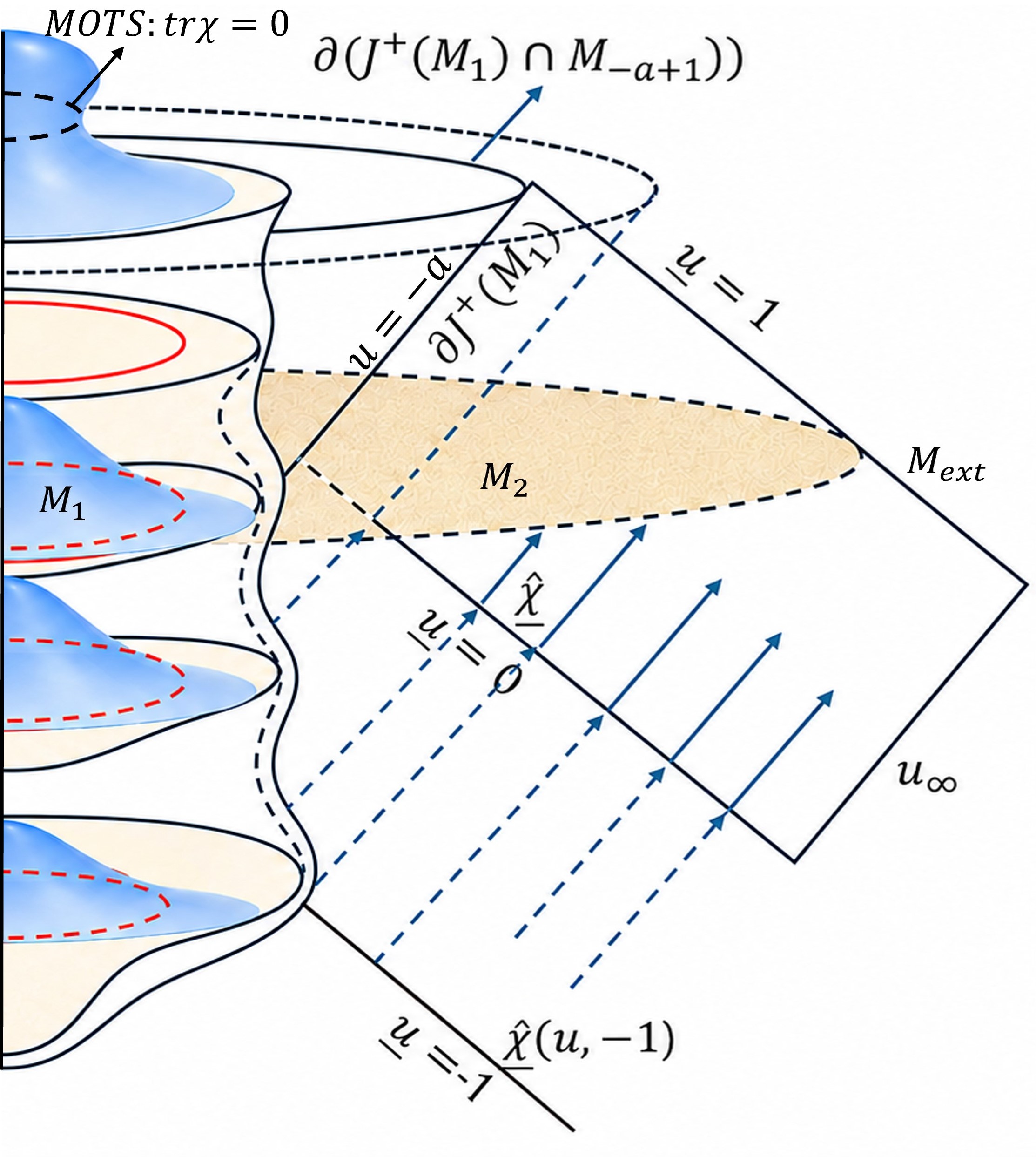}
 \caption{Figure schematically illustrates the boundary-geometry-driven MOTS formation in vacuum. It depicts the evolution of a compact domain whose boundary geometry reaches the Yau threshold, leading to MOTS formation.}
 \label{evol1}
\end{figure} 

\begin{theorem}[]
\label{thm:interior-mots-criterion}
Let $(\mathcal{U},\mathbf{g})$ be a smooth vacuum spacetime written in
spacetime-harmonic coordinates $(t,x^1,x^2,x^3)$:
$\mathbf{g}^{\alpha\beta}\Gamma^\mu_{\alpha\beta}=0.$
Write the spacetime metric in ADM form as
\begin{equation}
\mathbf{g}
=
-N^2dt^2
+
g_{ij}
\left(dx^i+X^i\,dt\right)
\left(dx^j+X^j\,dt\right),
\label{eq:spacetime-harmonic-adm}
\end{equation}
where $N$ is the lapse and $X=X^i\partial_i$ is the shift.
Assume that
$0<N_-\leq N\leq N_+$
on the spacetime domain under consideration.
Let $M_0\subset\mathcal M_0$ be a compact domain with smooth boundary
$S_0:=\partial M_0$, and assume that $M_0$ contains no MOTS.
Let $\Phi_t:M_0\longrightarrow\mathcal M_t$ be the flow generated by the shift-corrected time vector field,
$\frac{d}{dt}\Phi_t(x)
=
(\partial_t-X)_{\Phi_t(x)},
\qquad
\Phi_0(x)=x.$ Define
$M_t:=\Phi_t(M_0),
\qquad
S_t:=\Phi_t(S_0).$
Since
$\partial_t-X=Ne_T,$
the domains $M_t$ are transported in the future normal direction.
Let $e_S$ denote the outward unit normal to $S_t$ within
$\mathcal M_t$, and set
$L:=e_T+e_S,
\qquad
\Lbar:=e_T-e_S.$
Assume that, for every $t\in[0,t_1]$,
$\tr\chi>0,
\qquad
\tr\chibar<0
\qquad
\text{on }S_t.$
Define the boundary growth functionals
\begin{eqnarray}
\mathfrak G_-
&:=&
-\frac12\tr\chi
\left(\tr\chi+\tr\chibar\right)
-
|\hat\chi|_\Sigma^2
+
2(\underline\omega-\omega)\tr\chi
\nonumber\\
&&
+
2\operatorname{div}_\Sigma\eta
+
2|\eta|_\Sigma^2
+
2\rho
-
\hat\chi\cdot\hat\chibar,
\nonumber\\
\mathfrak G_+
&:=&
\frac12\tr\chibar
\left(\tr\chi+\tr\chibar\right)
+
|\hat\chibar|_\Sigma^2
+
2(\underline\omega-\omega)\tr\chibar
\nonumber\\
&&
-
2\operatorname{div}_\Sigma\underline\eta
-
2|\underline\eta|_\Sigma^2
-
2\rho
+
\hat\chi\cdot\hat\chibar,
\label{eq:boundary-growth-functional}
\end{eqnarray}
where $\Sigma$ is the metric induced on $S_t$.
Assume that one of the following two phases persists throughout
$t\in[0,t_1]$.

\noindent During the contracting phase, $\kappa<0$ on $S_t$, and therefore the generalized mean curvature is given by $c(t)=\min_{S_t}\tr\chi$. Let $\mathcal P_t:={p\in S_t:\tr\chi(t,p)=c(t)}$ denote the set of minimizing points of $\tr\chi$ on $S_t$, and assume that $\inf_{p\in\mathcal P_t}\mathfrak G_-(t,p)\geq2\gamma_0$ for every $t\in[0,t_1]$. During the expanding phase, $\kappa>0$ on $S_t$, and hence $c(t)=-\max_{S_t}\tr\chibar$. Let $\mathcal Q_t:={q\in S_t:\tr\chibar(t,q)=\max_{S_t}\tr\chibar(t,\cdot)}$ denote the set of maximizing points of $\tr\chibar$ on $S_t$.
and assume that
$\inf_{q\in\mathcal Q_t}
\mathfrak G_+(t,q)
\geq
2\gamma_0
\qquad
\text{for every }t\in[0,t_1].$
Let $\varepsilon_{\mathrm{def}}(t)
:=
\int_0^t
\|Nk\|_{L^\infty(M_\tau,g(\tau))}
\,d\tau.$
Assume that
$\operatorname{Rad}(M_0)c(0)
<
\frac{3\pi}{2},$
and that
\begin{eqnarray}
    e^{-\varepsilon_{\mathrm{def}}(t_1)}
\operatorname{Rad}(M_0)
\left(
c(0)+N_-\gamma_0t_1
\right)
\geq
\frac{3\pi}{2}.
\label{eq:quantitative-crossing}
\end{eqnarray}
Then $M_t$ contains a MOTS for some $t\in(0,t_1]$.
\end{theorem}

The proof of Theorem~\ref{thm:interior-mots-criterion} is given in Appendix~\ref{app:proof-thm1}. Now we explain the physics of the main theorem.
This may occur through an increase of \(c(t)\), sufficiently weak variation of the \(H\)-radius, or a combination of both effects.
Because \(\partial M_t\) is compact, the minimum of \(\tr\chi\) or the maximum of \(\tr\chibar\) are attained. As defined in Theorem \ref{thm:interior-mots-criterion},
$\mathcal P_t
:=
\left\{
p\in\partial M_t:
\tr\chi(p)=c(t)
\right\}$, $\mathcal Q_t
:=
\left\{
q\in\partial M_t:
-\tr\chibar(q)=c(t)
\right\}$, and we further define $\mathcal L_t
:=
\left\{
l\in\partial M_t:
H(l)=c(t)
\right\}$
which are the sets of minimizing points. After identifying the surfaces \(\partial M_t\) by the \(e_T\)-flow, the one-sided derivative of \(c\) is determined by the evolution of \(\tr\chi\), \(\tr\chibar\), and $H$ on \(\mathcal P_t\), \(\mathcal Q_t\), and $\mathcal L_t$ respectively. In particular, for contraction phase, i.e., $\kappa<0$, when the minimizer \(p_t\) is unique and varies smoothly,
$\frac{dc}{dT}(t)=
\nabla_{e_T}\tr\chi(p_t).$
Hence, $\nabla_{e_T}\tr\chi(p_t)>0$
implies that the minimum increases locally. More generally, if the minimum is attained at several points, it is sufficient to require
$\inf_{p\in\mathcal P_t}
\nabla_{e_T}\tr\chi(p)>0$, which as shown in Appendix \ref{app:proof-thm1} reduces to $\inf_{p\in\mathcal P_t}
\mathfrak{G}_->0$. For the expansion and area-preserving shear phases, the requirements are $\inf_{q\in\mathcal Q_t}
\mathfrak{G}_+>0$ and $\inf_{l\in\mathcal L_t}
(\mathfrak{G}_++\mathfrak{G}_-)>0)$, respectively.
However, mere positivity of the corresponding growth functional is not sufficient to guarantee that the Yau threshold is crossed within the available evolution interval. Theorem~\ref{thm:interior-mots-criterion} therefore imposes the quantitative lower bound $\inf_{\mathcal P_t}\mathfrak G_-\geq2\gamma_0$ in the contracting phase or $\inf_{\mathcal Q_t}\mathfrak G_+\geq2\gamma_0$ in the expanding phase. Together with the controlled deformation of the $H$-radius, this ensures that $\operatorname{Rad}(M_t)c(t)$ reaches $3\pi/2$ for some $t\in(0,t_1]$ (see Appendix~\ref{app:proof-thm1}). The following Theorem~\ref{thm2} verifies this quantitative growth condition in the Theorem~\ref{thm:interior-mots-criterion}. 
\begin{theorem}
\label{thm2}
Let $a\gg1$ and fix $0<\delta<\frac14$. Let $(\mathcal U_a,\mathbf g_a)$ be the vacuum development of an asymptotically flat smooth initial data verifying the estimates in Appendix \ref{app3} in spacetime-harmonic coordinates. Let $\Omega_a:=J^+\!\left(M_{\mathrm{int}}^{1/a}\right)\cap\{-a\leq t\leq-a+t_1\}$, where $0<t_1\leq\frac34$.
Introduce the shifted evolution parameter $\tau:=t+a\in[0,t_1]$. Let $\Phi_\tau$ denote the flow generated by $\partial_t-X=Ne_T$, starting from the initial slice $t=-a$, so that $\frac{d}{d\tau}\Phi_\tau=(\partial_t-X)\circ\Phi_\tau$ and $\Phi_0=\mathrm{Id}$. Let $M_0$ be a compact domain on the initial slice and define $M_\tau:=\Phi_\tau(M_0)$ and $S_\tau:=\partial M_\tau=\Phi_\tau(S_0)$.
Assume that $M_0$ contains no MOTS and that the propagated interior Cauchy estimates in $\Omega_a$ given in Appendix~\ref{app3} hold. Let $\psi:=\Ric^{(3)}-\frac13\Scal^{(3)}g$.
Assume that there exist constants $\mu_*>0$ and $\kappa_*>0$, independent of $a$, such that one of the following two initial alternatives holds on $S_0$.
In the contracting alternative, assume that $\kappa(0,p)\leq-\frac{\mu_*}{a^{3/2}}$ for every $p\in S_0$, and that
$$\inf_{p\in S_0}\psi(e_{S_0},e_{S_0})(0,p)\geq\frac{2\kappa_*}{a^2},$$
where $e_{S_0}$ denotes the outward unit normal to $S_0$ in the initial Cauchy hypersurface.
In the expanding alternative, assume that $\kappa(0,p)\geq\frac{\mu_*}{a^{3/2}}$ for every $p\in S_0$, and that
$$\sup_{p\in S_0}\psi(e_{S_0},e_{S_0})(0,p)\leq-\frac{2\kappa_*}{a^2}.$$
Define $\gamma_a:=\frac{\kappa_*}{2a^2}$. Assume moreover that $c(0)>0$, $\Rad(M_0)c(0)<\frac{3\pi}{2}$, and that
$\exp\!\left(-C_\delta t_1a^{-3/2+\delta}\right)\Rad(M_0)\left[c(0)+\left(1-C_\delta a^{-3/2+\delta}\right)\gamma_a t_1\right]$ $\geq\frac{3\pi}{2}$.
Then, for all sufficiently large $a$, the corresponding initial phase and anisotropy conditions persist throughout $\tau\in[0,t_1]$. More precisely, in the contracting phase,
$\kappa(\tau,p)<0$ on $S_\tau$ and
$\inf_{p\in S_\tau}\psi(e_S,e_S)(\tau,p)>\frac{\kappa_*}{a^2}$,
whereas in the expanding phase,
$\kappa(\tau,p)>0$ on $S_\tau$ and
$\sup_{p\in S_\tau}\psi(e_S,e_S)(\tau,p)<-\frac{\kappa_*}{a^2}$.
Consequently,
$\inf_{p\in\mathcal P_\tau}\mathfrak G_-(\tau,p)\geq2\gamma_a$
throughout the contracting phase, or
$\inf_{q\in\mathcal Q_\tau}\mathfrak G_+(\tau,q)\geq2\gamma_a$
throughout the expanding phase.
Hence the corresponding hypotheses of Theorem~\ref{thm:interior-mots-criterion} are satisfied.
Therefore, there exists $\tau_*\in(0,t_1]$ such that $\Rad(M_{\tau_*})c(\tau_*)\geq\frac{3\pi}{2}$, and $M_{\tau_*}$ contains a MOTS.
\end{theorem}
\noindent First note that the initial conditions on $\kappa$ and $\psi(e_{S_0},e_{S_0})$ are
compatible with the Cauchy data in
Appendix \ref{app3}, namely
$|\kappa|=O(a^{-3/2})$ and $|\psi|=O(a^{-2})$. Appendix~\ref{app:proof-thm2} shows that these
signs, together with the corresponding contracting or expanding phase,
persist throughout the semi-global evolution and yield the quantitative
growth condition required in
Theorem~\ref{thm:interior-mots-criterion}. They therefore provide a
sufficient, but not necessary, intrinsic mechanism for reaching Yau's
threshold.

The persistent area-preserving shear phase is different. In this case $\kappa=0$ throughout the evolution interval, so $c(t)=\min_{S_t}H$ and $\nabla_{e_T}\kappa=0$. Since $\mathfrak G_--\mathfrak G_+=4\nabla_{e_T}\kappa$, one has $\mathfrak G_-=\mathfrak G_+$, and combining this with the above leading-order reductions forces $\psi(e_S,e_S)=O(a^{-5/2})$. Hence the $O(a^{-2})$ contribution that drives the growth of $c(t)$ in the contracting and expanding phases is absent. Defining $\mathfrak G_{\mathrm{0}}:=\frac12(\mathfrak G_-+\mathfrak G_+)=2\nabla_{e_T}H$, one obtains $\mathfrak G_{\mathrm{0}}=O(a^{-5/2})$, and therefore $\frac{d}{dt}c(t)=O(a^{-5/2})$. Over the available $O(1)$ time interval this yields only $\Delta c=O(a^{-5/2})$; since $\operatorname{Rad}(M_t)=O(a)$, the corresponding change in $\operatorname{Rad}(M_t)c(t)$ is only $O(a^{-3/2})$, up to the perturbatively small $H$-radius deformation controlled by Appendix~\ref{app3}. Therefore, the present semi-global estimates do not provide sufficient quantitative growth to guarantee that Yau's threshold is crossed during a persistent area-preserving shear phase, although they do not exclude MOTS formation by some other mechanism.

The tensor $\psi:=\operatorname{Ric}^{(3)}-\frac13\operatorname{Scal}^{(3)}g$ defined in Theorem~\ref{thm2} measures the intrinsic three-dimensional curvature anisotropy of $M_t$: it is the trace-free part of the three-dimensional Ricci tensor and therefore measures the deviation of the principal Ricci curvatures from their isotropic average, as explained in Appendix~\ref{app4}.
The intrinsic anisotropy enters the deformation of the spatial geometry through the ADM evolution equations
\begin{eqnarray}
\partial_t g_{ij}
&=&
-2Nk_{ij}
+
\left(\mathcal{L}_{X}g\right)_{ij},
\label{adm-metric-evolution}\\
\left(
\partial_t-\mathcal{L}_{X}
\right)k_{ij}
&=&
-\nabla_i\nabla_jN
+
N
\left[
k_i{}^{m}k_{mj}
-
\left(\operatorname{tr}_{g}k\right)k_{ij}
\right].\nonumber\\
&& ~~~~~~~~+N~\operatorname{Ric}^{(3)}_{ij}
\label{adm-k-evolution}
\end{eqnarray}
Taking the trace-free part of Eq.~(\ref{adm-k-evolution}) gives

$\left[
\left(
\partial_t-\mathcal{L}_{X}
\right)k
\right]^{\mathrm{TF}}
=
-\left(\nabla^{2}N\right)^{\mathrm{TF}}
+
N\psi
+
N\left(k\circ k\right)^{\mathrm{TF}}
-
N\left(\operatorname{tr}_{g}k\right)k^{\mathrm{TF}},$
so that $\psi$ is the leading intrinsic-curvature source for the trace-free evolution of $k$. Thus, while the sign of $\psi(e_S,e_S)$ controls the leading growth of $c(t)$, the full tensor $\psi$ governs the anisotropic part of the spatial deformation.
Now, let $\Gamma\subset M_t$ be an admissible spacelike curve. Its length is
$\ell_t(\Gamma)=\int_{\Gamma}\left(g_{ij}(t)\dot{x}^{i}\dot{x}^{j}\right)^{1/2}\,d\lambda$. Consequently, uniform bilipschitz control of $g(t)$ relative to the initial metric gives corresponding control of $\operatorname{Rad}(M_t)$. In the Appendix~(\ref{app:proof-thm1}), we show that the bilipschitz stability of the $H$-radius gives
\begin{equation}
\operatorname{Rad}\!\left(\Phi_T(M_{\mathrm{int}}^{1/a})\right)
\geq
\left(
1-C_{\delta}a^{-3/2+\delta}
\right)
\operatorname{Rad}\!\left(M_{\mathrm{int}}^{1/a}\right).
\label{radius-lower-bound}
\end{equation}
Since $\operatorname{Rad}(M_{\mathrm{int}}^{1/a})=O(a)$, its possible absolute decrease over this interval is only $O(a^{-1/2+\delta})$.
Thus, Theorems~\ref{thm:interior-mots-criterion} and~\ref{thm2} identify a unified geometric mechanism: the signed normal component $\psi(e_S,e_S)$ produces the required growth of $c(t)$ in either the contracting or expanding phase, while the propagated Cauchy hierarchy in $\Omega_a$ from Appendix~\ref{app3} keeps the deformation of the $H$-radius perturbatively small. Consequently, for some $t\in(0,t_1]$, the product $\operatorname{Rad}(M_t)c(t)$ can cross Yau's threshold $\frac{3\pi}{2}$, at which point a MOTS is forced to occur inside $M_t$.

\noindent The trace-free Ricci tensor $\psi$ also connects the interior mechanism
described above with the characteristic construction of
Ref.~\cite{MondalYau2026}. In the present formulation, the hypotheses are
imposed on the evolving compact domain $M_t$, whereas
Ref.~\cite{MondalYau2026} additionally prescribes characteristic data on
the outgoing null boundary $\partial J^{+}(M_1)$. We now show that the
intrinsic anisotropy of $M_t$ contributes directly to the extreme null curvature components (i.e., $\alpha$, $\underline{\alpha}$) governing that characteristic evolution.
In the Cauchy development
$\Omega
=
J^{+}\!\left(M_{\mathrm{int}}^{1/a}\right)
\cap
\left\{-a\leq t\leq -a+1\right\},$
the propagated estimates discussed before give
$\alpha_{AB}
=
2\psi^{TT}_{AB}
+
O\left(a^{-5/2}\right),~
\underline{\alpha}_{AB}
=
2\psi^{TT}_{AB}
+
O\left(a^{-5/2}\right),$ where $\psi^{TT}_{AB}
:=
\left(
e_A{}^{i}e_B{}^{j}\psi_{ij}
\right)^{\mathrm{TF}_{\Sigma}}$, where $\mathrm{TF}_{\Sigma}$ denotes the trace-free projection with respect to induced metric
$\Sigma$ on $S_t$ and $TT$ denotes tangential-trace-free projection. Therefore, at the Cauchy--characteristic matching interface, the $\psi^{TT}$ supplies the leading anisotropic contribution to the
extreme null curvature components.
Now, in the characteristic development attached to
$M_{\mathrm{int}}^{1/a}\setminus M_1$, the extreme null curvature
components act as sources for the null shears. With the conventions of
Ref.~\cite{MondalYau2026}, the relevant null structure equations are
$\nabla_L \widehat{\chi}_{AB}
=
-\left(\operatorname{tr}\chi\right)\widehat{\chi}_{AB}
-2\omega\widehat{\chi}_{AB}
-\alpha_{AB},~
\nabla_{\underline{L}}\widehat{\underline{\chi}}_{AB}
=
-\left(\operatorname{tr}\underline{\chi}\right)
\widehat{\underline{\chi}}_{AB}
-2\underline{\omega}\widehat{\underline{\chi}}_{AB}
-\underline{\alpha}_{AB}.$  Thus, $\alpha$ and $\underline{\alpha}$ source the outgoing and incoming
shears, respectively. The incoming shear can then be controlled by integrating
the above null structure equation for $\widehat{\underline{\chi}}$ along the integral curves of
$\underline{L}$. More precisely, let
$\gamma:[s_0,s_1]\longrightarrow D_{a,-1}$ be an integral curve of
$\underline{L}$, with $\gamma(s_0)=q_0$ and $\gamma(s_1)=q_1$. A
standard transport estimate gives
$\left|
\widehat{\underline{\chi}}
\right|(q_1)
\leq
\exp\left[
C\int_{s_0}^{s_1}
\left(
\left|\tr\underline{\chi}\right|
+
\left|\underline{\omega}\right|
\right)(\gamma(s))\,ds
\right]
\nonumber\\
\quad\times
\left[
\left|
\widehat{\underline{\chi}}
\right|(q_0)
+
\int_{s_0}^{s_1}
\left|
\underline{\alpha}
\right|(\gamma(s))\,ds
\right].$
Thus, the shear induced at the matching interface remains quantitatively
controlled throughout the characteristic development. In particular,
the characteristic shear data are compatible with the smoothly matched
interior Cauchy data. They should not, however, be described as uniquely
generated by the interior data, since part of the characteristic data is
prescribed independently.

Thus, the same anisotropic tensor $\psi$ organizes both the interior and characteristic descriptions of the mechanism. In the Cauchy surface, Theorem~\ref{thm2} shows that an initial signed anisotropy of order $a^{-2}$ persists throughout the semi-global evolution: $\psi(e_S,e_S)>0$ drives the required growth of $c(\tau)$ in the contracting phase, whereas $\psi(e_S,e_S)<0$ does so in the expanding phase. At the same time, $\kappa(\tau)=\kappa(0)+O_\delta(a^{-2+\delta})$, so an initial contraction or expansion of order $a^{-3/2}$ preserves its sign on $\tau\in[0,t_1]$. Hence Yau's threshold can be reached through either contracting or expanding boundary evolution. An expanding outer domain may also contain a contracting compact subdomain to which the same criterion applies independently. The full tensor $\psi$ provides the leading intrinsic-curvature source for the trace-free evolution of $k$, while the propagated Cauchy hierarchy keeps the $H$-radius nearly preserved. At the matching interface, $\psi^{TT}$ gives the leading anisotropic contribution to $\alpha$ and $\underline{\alpha}$, which source the null shears. The interior and characteristic descriptions are therefore coupled manifestations of the same anisotropic vacuum dynamics, and the characteristic data of Ref.~\cite{MondalYau2026} may be viewed as its null-boundary realization.
This matching is illustrated in Fig.~\ref{evol1}. The compact domain is extended to the past through a collar reaching $\underline u=-1$, where the Cauchy data are smoothly matched to the characteristic development in $D_{a,-1}$ determined by the data on $\underline u=0$ and $u=u_\infty$. The same Cauchy estimates hold in the transported collar $M_\tau\setminus\widetilde M_\tau$, and the induced data agree smoothly across the matching interface. Yau's boundary criterion is therefore compatible with a broad class of dynamical compact domains and is not intrinsically tied to contraction.

Several questions remain open. Ref.~\cite{MondalYau2026} reaches the first MOTS but does not describe its subsequent evolution. In the present setting, $\operatorname{tr}\underline{\chi}<0$ at the first MOTS, and a generic first MOTS is strictly stable. If $\lambda_1>0$ and $\varphi>0$ are the principal eigenvalue and eigenfunction of the stability operator $\mathcal L_{\mathrm{MOTS}}$, then an inward deformation with velocity $-\varphi\nu$ gives
$\left.\frac{d}{d\varepsilon}\operatorname{tr}\chi\right|_{\varepsilon=0}
=-\lambda_1\varphi<0$.
Hence, for sufficiently small $\varepsilon>0$, the deformed surfaces satisfy $\operatorname{tr}\chi<0$ and $\operatorname{tr}\underline{\chi}<0$, and are therefore trapped. Under the standard hypotheses of Penrose's theorem, this implies future null geodesic incompleteness. The degenerate case $\lambda_1=0$ requires a separate rigidity analysis. Thus, the main remaining issue is the evolution of the resulting marginally trapped tube and trapped surface. Other directions include extending the mechanism to scale-critical regularity, possibly related to $\dot H^{3/2}$ control. Finally, the same
geometric mechanism may be studied for nonvacuum systems, where the
matter fields contribute to both the constraint equations and the null
focusing equations. Such an extension could provide a geometric framework for black-hole formation in physically realistic matter configurations and other strong-field regimes.

\appendix

\section{Proof of the Theorem~\ref{thm:interior-mots-criterion}}
\label{app:proof-thm1}

\begin{proof}With the null-frame conventions used here, the outgoing Raychaudhuri
equation and the cross-focusing equation are
\begin{align}
\nabla_L\tr\chi
&=
-\frac{1}{2}
\left(
\tr\chi
\right)^2
-
|\hat{\chi}|_{\Sigma}^{2}
-
2\omega\tr\chi,
\label{eq:raychaudhuri-theorem}\\
\nabla_{\Lbar}\tr\chi
&=
-\frac{1}{2}
\tr\chibar\,\tr\chi
+
2\underline{\omega}\tr\chi
+
2\operatorname{div}_{\Sigma}\eta
+
2|\eta|_{\Sigma}^{2}
\nonumber\\
&\quad
+
2\rho
-
\hat{\chi}\cdot\hat{\chibar}.
\label{eq:cross-focusing-theorem}
\end{align}
Since
$e_T
=
\frac{1}{2}
\left(
L+\Lbar
\right)$,
adding Eqs.~(\ref{eq:raychaudhuri-theorem}) and
(\ref{eq:cross-focusing-theorem}) gives
$2\nabla_{e_T}\tr\chi
=
\mathfrak{G}_-.$
Since $\Phi_t$ is generated by
$\partial_t-X=Ne_T$, and
$c(t)
=
\min_{p\in S_0}\tr\chi
\left(
t,\Phi_t(p)
\right),$
by the standard envelope formula for the minimum of a smooth
family of functions,
$D^+c(t)
=
\frac{1}{2}
\inf_{q\in\mathcal P_t}
N(t,q)\mathfrak{G}_-(t,q),$
where $\mathcal P_t$ is the set of minimizing points of $\tr\chi$ on $S_t$.
Since $0<N_-\leq N\leq N_+$ and
$\inf_{q\in\mathcal P_t}\mathfrak{G}_-(t,q)
\geq
2\gamma_0$, we obtain
$D^+c(t)\geq N_-\gamma_0.$
Integrating from $0$ to $t$ gives
\begin{equation}
c(t)
\geq
c(0)
+
N_-\gamma_0t.
\label{eq:c-linear-growth}
\end{equation}
This argument does not require the minimizing point of $\tr\chi$ to be
unique or to vary smoothly with time.

\noindent
Similarly, for the expansion phase, the useful identity is
$\mathfrak{G}_+
=
-2\nabla_{e_T}\tr\chibar$,
which can be obtained from the null structure equations as above.
Since for this scenario,
$c(t)
=
-\max_{S_t}\tr\chibar~,$
the corresponding envelope formula gives
$D^+c(t)
=
\frac{1}{2}
\inf_{q\in\mathcal Q_t}
N(t,q)\mathfrak{G}_+(t,q),$
where $\mathcal Q_t$ is the set of maximizing points of
$\tr\chibar$ on $S_t$.
Since $0<N_-\leq N\leq N_+$ and
$\inf_{q\in\mathcal Q_t}\mathfrak{G}_+(t,q)
\geq
2\gamma_0$, we again obtain
$D^+c(t)\geq N_-\gamma_0$,
and hence
$c(t)\geq
c(0)+N_-\gamma_0t$.

We next control the $H$-radius. Pulling the spatial metric back by the
flow $\Phi_t$ gives
\begin{align}
\frac{d}{dt}
\left(
\Phi_t^*g(t)
\right)
=
\Phi_t^*
\left[
\partial_tg-\mathcal{L}_Xg
\right]
=
-2\Phi_t^*(Nk).
\label{eq:pulled-back-metric-evolution}
\end{align}
Let $v\in T_pM_0$ be nonzero. From
Eq.~(\ref{eq:pulled-back-metric-evolution}),

$\left|
\frac{d}{dt}
\log
\left[
\Phi_t^*g(t)(v,v)
\right]
\right|
=
2
\frac{
\left|
\Phi_t^*(Nk)(v,v)
\right|
}{
\Phi_t^*g(t)(v,v)
}
\leq
2
\|Nk\|_{L^\infty(M_t,g(t))}.$
Integration yields
$e^{-2\varepsilon_{\mathrm{def}}(t)}g(0)
\leq
\Phi_t^*g(t)
\leq
e^{2\varepsilon_{\mathrm{def}}(t)}g(0).$
Consequently, the corresponding distance functions satisfy
\begin{equation}
e^{-\varepsilon_{\mathrm{def}}(t)}
d_{g(0)}
\leq
d_{\Phi_t^*g(t)}
\leq
e^{\varepsilon_{\mathrm{def}}(t)}
d_{g(0)}.
\label{eq:distance-comparison}
\end{equation}
Let $\Gamma\subset M_0$ be an admissible curve and let
$\Gamma_t:=\Phi_t(\Gamma)$. If
$0<r<\operatorname{Rad}(M_0,\Gamma)$,
then
$\operatorname{dist}_{g(0)}
\left(
\Gamma,\partial M_0
\right)
>
r$,
and $\Gamma$ does not bound a disk in
$N_r^{g(0)}(\Gamma)$. By
Eq.~(\ref{eq:distance-comparison}),
$\operatorname{dist}_{g(t)}
\left(
\Gamma_t,\partial M_t
\right)
>
e^{-\varepsilon_{\mathrm{def}}(t)}r.$
Moreover,
$\Phi_t^{-1}
\left(
N_{e^{-\varepsilon_{\mathrm{def}}(t)}r}^{g(t)}
(\Gamma_t)
\right)
\subset
N_r^{g(0)}(\Gamma).$
Hence, $\Gamma_t$ does not bound a disk in
$N_{e^{-\varepsilon_{\mathrm{def}}(t)}r}^{g(t)}
(\Gamma_t).$
Taking the supremum first over $r$ and then over admissible curves
$\Gamma$ gives
\begin{equation}
\operatorname{Rad}(M_t)
\geq
e^{-\varepsilon_{\mathrm{def}}(t)}
\operatorname{Rad}(M_0).
\label{eq:h-radius-lower-bound}
\end{equation}
Now combining Eqs.~(\ref{eq:c-linear-growth}) and
(\ref{eq:h-radius-lower-bound}), we obtain
$$
\operatorname{Rad}(M_t)c(t)
\geq
e^{-\varepsilon_{\mathrm{def}}(t)}
\operatorname{Rad}(M_0)
\left(
c(0)+N_-\gamma_0t
\right).
$$
At $t=t_1$, Eq.~(\ref{eq:quantitative-crossing}) implies
$\operatorname{Rad}(M_{t_1})c(t_1)
\geq
\frac{3\pi}{2}.$
Yau's vacuum boundary theorem therefore implies that $M_{t_1}$
contains a MOTS.
If a MOTS forms at an earlier time, the conclusion already holds.
Otherwise, $t_1$ is a formation time. Thus, in either case, $M_t$
contains a MOTS for some $t\in(0,t_1]$.
\end{proof}

\section{Proof of the Theorem~\ref{thm2}}
\label{app:proof-thm2}

\begin{proof}
We verify that, for sufficiently large $a$, the corresponding hypotheses of Theorem~\ref{thm:interior-mots-criterion} are satisfied. Using the null Gauss equation together with the Gauss equation for $S_\tau\subset\mathcal M_{-a+\tau}$, the boundary growth functionals admit the exact Cauchy representations
\begin{align}
\mathfrak G_+
={}&
-2\psi(e_S,e_S)
+
\frac13\operatorname{Scal}^{(3)}
-
2\hat s\cdot\hat k_{\Sigma}
+
|\hat k_{\Sigma}|^2
-
H\kappa\nonumber\\
&\quad+
\frac32\kappa^2
+
2(\underline\omega-\omega)\tr\chibar
-
2\operatorname{div}_{\Sigma}\underline\eta
-
2|\underline\eta|_{\Sigma}^2,
\label{eq:Gplus-exact-Cauchy-proof}
\\
\mathfrak G_-
={}&
2\psi(e_S,e_S)
-
\frac13\operatorname{Scal}^{(3)}
-
2\hat s\cdot\hat k_{\Sigma}
-
|\hat k_{\Sigma}|^2
-
H\kappa\nonumber\\
&\quad
-
\frac32\kappa^2
+
2(\underline\omega-\omega)\tr\chi
+
2\operatorname{div}_{\Sigma}\eta
+
2|\eta|_{\Sigma}^2.
\label{eq:Gminus-exact-Cauchy-proof}
\end{align}

\noindent The propagated Cauchy hierarchy and the corresponding boundary estimates give
$|\operatorname{Ric}^{(3)}|\lesssim a^{-2+\delta}$,
$|\operatorname{Scal}^{(3)}|\lesssim a^{-3+2\delta}$,
$|\psi|\lesssim a^{-2+\delta}$,
$|s|+|H|\lesssim a^{-1}$, and
$|\kappa|+|\hat k_{\Sigma}|\lesssim a^{-3/2+\delta}$.
Consequently,
$|\hat s\cdot\hat k_{\Sigma}|\lesssim a^{-5/2+\delta}$,
$|\hat k_{\Sigma}|^2+\kappa^2\lesssim a^{-3+2\delta}$, and
$|H\kappa|\lesssim a^{-5/2+\delta}$.
Moreover,
$|(\underline\omega-\omega)\tr\chi|
+
|(\underline\omega-\omega)\tr\chibar|
\lesssim a^{-7/2+\delta}$,
$|\operatorname{div}_{\Sigma}\eta|
+
|\operatorname{div}_{\Sigma}\underline\eta|
\lesssim a^{-5/2+\delta}$, and
$|\eta|_{\Sigma}^2+
|\underline\eta|_{\Sigma}^2
\lesssim a^{-3+2\delta}$.
Using these estimates, Eqs.~\eqref{eq:Gplus-exact-Cauchy-proof} and
\eqref{eq:Gminus-exact-Cauchy-proof} reduce to,
\begin{equation}
\mathfrak G_-
=
2\psi(e_S,e_S)+\mathcal R_-,
\qquad
\mathfrak G_+
=
-2\psi(e_S,e_S)+\mathcal R_+,
\label{eq:Gpm-leading-proof}
\end{equation}
where, uniformly for $\tau\in[0,t_1]$,
$|\mathcal R_-|+|\mathcal R_+|
\leq
C_\delta a^{-5/2+\delta}
=
o(a^{-2}).$
Now, the differentiated Cauchy hierarchy gives
$\left|
\frac{d}{d\tau}
\left[
\psi(e_S,e_S)\bigl(\tau,\Phi_\tau(p)\bigr)
\right]
\right|
\leq
C_\delta a^{-7/2+2\delta},$
and hence, integrating for $O(1)$ time interval we obtain
\begin{equation}
\psi(e_S,e_S)\bigl(\tau,\Phi_\tau(p)\bigr)
=
\psi(e_{S_0},e_{S_0})(0,p)
+
O_\delta(a^{-7/2+2\delta}).
\label{eq:psi-persistence-proof}
\end{equation}
Therefore, the contracting initial condition implies, for sufficiently large $a$,
$\inf_{p\in S_\tau}
\psi(e_S,e_S)(\tau,p)
>
\frac{\kappa_*}{a^2},$
whereas the expanding initial condition implies
$\sup_{p\in S_\tau}
\psi(e_S,e_S)(\tau,p)
<
-\frac{\kappa_*}{a^2}.$

We next verify persistence of the sign of $\kappa$. Using the above estimates on $\mathfrak G_\pm$, we obtain
$\nabla_{e_T}\kappa
=
\frac14(\mathfrak G_--\mathfrak G_+)
=
\psi(e_S,e_S)
+
O_\delta(a^{-5/2+\delta}).$
Along the flow $\Phi_\tau$, the propagated estimates therefore give
$$\kappa\bigl(\tau,\Phi_\tau(p)\bigr)
=
\kappa(0,p)
+
O_\delta(a^{-2+\delta}).$$

Therefore, for large $a$, the initial condition
$\kappa(0,p)\leq-\mu_*a^{-3/2}$ implies
$\kappa(\tau,p)
\leq
-\frac{\mu_*}{2a^{3/2}}
\qquad
\text{on }S_\tau,$
whereas
$\kappa(0,p)\geq\mu_*a^{-3/2}$ implies
$\kappa(\tau,p)
\geq
\frac{\mu_*}{2a^{3/2}}
>0
\qquad
\text{on }S_\tau.$
Thus, the initial contracting or expanding phase persists throughout $\tau\in[0,t_1]$.
In the contracting phase, Eqs.~\eqref{eq:Gpm-leading-proof} along with $\inf_{p\in S_\tau}
\psi(e_S,e_S)(\tau,p)
>
\frac{\kappa_*}{a^2},$,  yield
$\inf_{p\in\mathcal P_\tau}\mathfrak G_-(\tau,p)
>
\frac{\kappa_*}{a^2}
=
2\gamma_a$
for sufficiently large $a$, where
$\gamma_a:=\frac{\kappa_*}{2a^2}$.
Similarly, in the expanding alternative,
$\inf_{q\in\mathcal Q_\tau}\mathfrak G_+(\tau,q)
>
\frac{\kappa_*}{a^2}
=
2\gamma_a .$
Since the signed bounds on $\psi(e_S,e_S)$ hold on the whole of $S_\tau$, the same estimates give $\mathfrak G\pm>0$ on $S_\tau$ in the corresponding phases. Hence the corresponding positive-growth
hypothesis of
Theorem~\ref{thm:interior-mots-criterion} also holds during the
contracting phase.

The propagated lapse estimate gives
$1-C_\delta a^{-3/2+\delta}\leq N\leq1+C_\delta a^{-3/2+\delta}$, so that $N_-=1-C_\delta a^{-3/2+\delta}>0$.
Theorem~\ref{thm:interior-mots-criterion} therefore yields
$c(\tau)
\geq
c(0)
+
\left(
1-C_\delta a^{-3/2+\delta}
\right)
\gamma_a\tau .$
The deformation budget satisfies
$\varepsilon_{\mathrm{def}}(\tau)
\leq
C_\delta\tau a^{-3/2+\delta}$.
Hence
$\operatorname{Rad}(M_\tau)
\geq
\exp\!\left(
-C_\delta\tau a^{-3/2+\delta}
\right)
\operatorname{Rad}(M_0),$
and consequently
\begin{align}
\operatorname{Rad}(M_\tau)c(\tau)
&\geq
\exp\!\left(
-C_\delta\tau a^{-3/2+\delta}
\right)
\operatorname{Rad}(M_0)
\nonumber\\
&\quad\times
\left[
c(0)
+
\left(
1-C_\delta a^{-3/2+\delta}
\right)
\gamma_a\tau
\right].
\label{eq:product-theorem2-proof}
\end{align}

At $\tau=t_1$, the hypothesis $\exp\!\left(-C_\delta t_1a^{-3/2+\delta}\right)\Rad(M_0)\left[c(0)+\left(1-C_\delta a^{-3/2+\delta}\right)\gamma_a t_1\right]$ $\geq\frac{3\pi}{2}$ gives
$\operatorname{Rad}(M_{t_1})c(t_1)\geq\frac{3\pi}{2}$.
Yau's vacuum boundary theorem therefore implies that $M_{t_1}$ contains a MOTS. If a MOTS forms earlier, the conclusion already holds. Hence there exists $\tau_*\in(0,t_1]$ such that $M_{\tau_*}$ contains a MOTS.
\end{proof}
\section{Estimates on $\widetilde{M}_1$}
\label{app3}
On the interior region, in a fixed
harmonic coordinate chart on $\widetilde{M}_1$, the Cauchy data constructed in
Ref.~\cite{MondalYau2026} satisfy the following hierarchy:
\begin{eqnarray}
&&\|\partial^{m+1}g(-a)\|_{L^\infty(\widetilde{\mathcal M}_{1})}
\lesssim a^{-m-1},
\label{eq:g-hierarchy}\\
&&\|\partial^m k(-a)\|_{L^\infty(\widetilde{\mathcal M}_{1})}
\lesssim a^{-m-3/2},
\label{eq:k-hierarchy}\\
&&\|\partial^m(N-1)\|_{L^\infty(\widetilde{\mathcal M}_{1})}
+
\|\partial^mX\|_{L^\infty(\widetilde{\mathcal M}_{1})}
\lesssim a^{-m-3/2},\nonumber\\
\label{eq:NX-hierarchy}
\end{eqnarray}
for every fixed admissible integer \(m\geq0\). We also assume the
corresponding uniformly-local Sobolev estimates
$\|\partial^{m+1}g(-a)\|_{H^{s-1}_{\mathrm{ul}}
(\widetilde{\mathcal M}_{1})}
\lesssim a^{-m-1},~
\|\partial^m k(-a)\|_{H^{s-1}_{\mathrm{ul}}
(\widetilde{\mathcal M}_{1})}
\lesssim a^{-m-3/2},$
for some fixed \(s\gg1\), together with the analogous
\(H^s_{\mathrm{ul}}\)-control of \(N-1\) and \(X\).
Thus the estimates are scale-covariant: each additional spatial
derivative gives one additional factor \(a^{-1}\). These estimates are propagated in
$\Omega
:=
J^{+}\!\left(M_{\mathrm{int}}^{1/a}\right)
\cap
\left\{-a\leq t\leq -a+1\right\},$
where
$M_{\mathrm{int}}^{1/a}
:=
M_1
\cup
\left[
M_2
\setminus
\left(
\left[
u_{\infty},-a-\frac{1}{a}
\right]
\times[0,1]\times S
\cap M_2
\right)
\right].$
For every fixed $\delta\in(0,1/2)$, one has
\begin{eqnarray}
&&\|\partial^{m+1}g(-a)\|_{L^\infty(\Omega)}
\lesssim a^{-m-1+\delta},
\label{eq:g-hierarchy}\\
&&\|\partial^m k(-a)\|_{L^\infty(\Omega)}
\lesssim a^{-m-3/2+\delta},
\label{eq:k-hierarchy}\\
&&\|\partial^m(N-1)\|_{L^\infty(\Omega)}
+
\|\partial^mX\|_{L^\infty(\Omega)}
\lesssim a^{-m-3/2+\delta}.\nonumber\\
\label{eq:NX-hierarchy2}
\end{eqnarray}
while in the collar \(M_1\setminus\widetilde{M}_1\) the second fundamental form $k$, lapse $N$, and the shift vector field $X$ verify 
\begin{eqnarray}
 ||\partial^{k}k||_{L^\infty\cap H^{s-1}_{\mathrm{ul}}(M_1\setminus\widetilde{M}_1)}\lesssim a^{-(k+\frac{3}{2})},\nonumber\\~||\partial^{k}(N-1)||_{L^\infty\cap H^{s-1}_{\mathrm{ul}}(M_1\setminus\widetilde{M}_1)}\lesssim a^{-(k+\frac{3}{2})},\nonumber\\~||\partial^{k}X||_{L^\infty\cap H^{s-1}_{\mathrm{ul}}(M_1\setminus\widetilde{M}_1)}\lesssim a^{-(k+\frac{3}{2})}   
\end{eqnarray}
with smooth interpolation to the Cauchy data induced by \(D_{a,1}\) on the
overlap through $M_1\setminus\widetilde{M}_1$.  Assume moreover that the induced metric, second fundamental form, and
all tangential derivatives up to order \(N\) agree across \(S_{-a,0}\). Complete the exterior by gluing to a Kerr end satisfying
\[
m_{\mathrm{ADM}}\sim a^{1/2},
\qquad
J=O(a),
\]
while preserving the vacuum constraints and the above Sobolev bounds.  Denote the
resulting global vacuum data by \((\widetilde g,\widetilde k)\).

\section{Trace-free part of $\mathrm{Ric}^{(3)}$}
\label{app4}
\noindent At each point, choose an orthonormal basis
$\{e_1,e_2,e_3\}$ satisfying
$\mathrm{Ric}^{(3)}(e_i,e_j)
=
\lambda_i\delta_{ij}.$
Since
$\mathrm{Scal}^{(3)}
=
\lambda_1+\lambda_2+\lambda_3,$
the eigenvalues of $\psi$ are
$\lambda_i-\bar{\lambda},
~~
\bar{\lambda}
:=
\frac{1}{3}
\left(
\lambda_1+\lambda_2+\lambda_3
\right),$
and hence
$\lvert\psi\rvert_g^{2}
=
\sum_{i=1}^{3}
\left(
\lambda_i-\bar{\lambda}
\right)^{2}.$
Thus, $\lvert\psi\rvert_g$ measures the deviation of the Ricci
eigenvalues from their average.
In three dimensions, the sectional curvature of the plane spanned by
$e_i$ and $e_j$ is
\begin{eqnarray}
\mathcal{K}_{ij}
=
\frac{1}{2}
\left(
\lambda_i+\lambda_j-\lambda_k
\right)
=
\frac{1}{2}\mathrm{Scal}^{(3)}-\lambda_k,
~
\{i,j,k\}=\{1,2,3\}.\nonumber\\
\label{sectional-ricci-relation}
\end{eqnarray}
Therefore, the Ricci eigenvalues are equal if and only if the sectional
curvatures are equal. Equivalently, $\psi=0$ precisely when the
three-dimensional metric has isotropic sectional curvature at that point.
Thus, $\lvert\psi\rvert_g$ provides a pointwise measure of the intrinsic
curvature anisotropy of $M_t$.
\end{document}